\documentclass[12pt,a4paper]{article}

\usepackage[utf8]{inputenc} 
\usepackage[T1]{fontenc}
\usepackage{url}
\usepackage{ifthen}
\usepackage{cite}
\usepackage[cmex10]{amsmath}
\usepackage{amssymb}
\usepackage{bm}
\usepackage{amsthm,amsfonts,ascmac}
\usepackage{cases}
\usepackage{mathtools}

\theoremstyle{plain}

\newtheorem{thm}{Theorem}[section]
\newtheorem{defi}[thm]{Definition}
\newtheorem{lem}[thm]{Lemma}

\newtheorem{fact}[thm]{Fact}

\begin{document}
\title{Permutation-Invariant Quantum Codes for Deletion Errors}

\author{
Taro Shibayama \thanks{
Department of Mathematics and Informatics,
Graduate School of Science,
Chiba University
1-33 Yayoi-cho, Inage-ku, Chiba City,
Chiba Pref., JAPAN, 263-0022
}
\and
Manabu Hagiwara \footnotemark[1]
}

\date{}
\maketitle

\begin{abstract}
This paper presents conditions for constructing permutation-invariant quantum codes for deletion errors and provides a method for constructing them.
Our codes give the first example of quantum codes that can correct two or more deletion errors.
Also, our codes give the first example of quantum codes that can correct both multiple-qubit errors and multiple-deletion errors.
We also discuss a generalization of the construction of our codes at the end.
\end{abstract}

\section{Introduction}
Quantum error-correcting codes are one of the essential factors for the development of quantum computers, which were first introduced by Shor in 1995 \cite{Peter1995}.
Since then, many codes have been constructed, e.g., CSS codes\cite{Robert1996}, stabilizer codes\cite{Gottesman1997}, surface codes\cite{Fowler2012}, etc.

Recently, in 2019, Leahy et al. provided a way to turn quantum insertion-deletion errors into errors that can be handled by conventional methods under the specific assumption \cite{Leahy2019}.
Since then, quantum deletion error-correction have attracted  a lot of attention from researchers.
In quantum information theory, quantum deletion error-correction is a problem of determining the quantum state in the entire quantum system from a quantum state in a partial system.
Therefore it is related to various topics, e.g., quantum erasure error-correcting codes \cite{Grassl1997}, quantum secret sharing \cite{Hillery1999}, purification of quantum state \cite{Hughston1993}, quantum cloud computing \cite{Biamonte2017}, etc.

The first quantum deletion codes under the general scenario was constructed in 2020 \cite{Nakayama20201}.
Since then, several quantum deletion codes have been constructed so far \cite{Hagiwara20202,Nakayama20202,Shibayama2020}.
However, research on quantum deletion codes is not yet sufficiently advanced.
All the quantum deletion codes that have already been found can correct only single-deletion errors.
In other words, quantum error-correcting codes that can correct two or more deletion errors have not been found to date.
Furthermore, in classical coding theory, the construction of codes that can correct both substitution and deletion errors has attracted much attention \cite{smagloy2020,song2020}; however, no such codes have yet been found in quantum coding theory.

In this paper, we consider permutation-invariant quantum codes, which are invariant under any permutation of the underlying particles.
Permutation-invariant quantum codes have been constructed using a variety of techniques\cite{ruskai2000,pollatsek2004,Ouyang2014,Ouyang2016,Ouyang2017,OC2019}.
Recently, permutation-invariant quantum codes have been explored for applications such as for quantum storage\cite{Ouyang2020} and robust quantum metrology\cite{OSM2019}.
These are expected to be applied to physically realistic scenarios\cite{WWG2019}.

This paper is organized as follows.
In Section \ref{pre}, some notations and definitions are introduced.
In Section \ref{delsection}, the deletion error-correcting codes as  quantum codes are defined and the main theorem is stated, and in Section \ref{codeconst}, we give the proof of the main theorem.
In Section \ref{example}, we present two families of our codes.
The first gives quantum codes that can correct any number of deletion errors, and the code can also correct multiple qubit errors.
The second gives quantum codes that can correct single-deletion errors that have been found so far; these are also included in our codes.
In Section \ref{generalize}, we consider a generalization of our codes.
Finally, this paper is summarized in Section \ref{conc}.


\section{Preliminaries}\label{pre}
Throughout this paper, we fix the following notation.
Let $N\geq2$ be an integer and $[N]\coloneqq\{1,2,\dots ,N\}$.
For a square matrix $M$ over 
the complex field $\mathbb{C}$, we denote by ${\rm Tr}(M)$ the sum of the diagonal elements of $M$.
Set $|0\rangle, |1\rangle\in\mathbb{C}^2$ as $|0\rangle\coloneqq(1,0)^\top$, $|1\rangle\coloneqq(0,1)^\top$, and $|\bm{x}\rangle\coloneqq|x_1\rangle\otimes|x_2\rangle\otimes\dots\otimes|x_N\rangle\in\mathbb{C}^{2\otimes N}$ for a bit sequence $\bm{x}=x_1x_2\cdots x_N\in\{0,1\}^N$.
Here $\otimes$ is the tensor product operation and $\top$ is the transpose operation.
Let $\langle\bm{x}|\coloneqq|\bm{x}\rangle^\dag$ denote the conjugate transpose of $|\bm{x}\rangle$.
We define the Hamming weight ${\rm wt}(\bm{x})$ of $\bm{x}$ by ${\rm wt}(\bm{x})\coloneqq\#\{p\in [N]\mid x_p\neq0\}$.

A positive semi-definite Hermitian matrix of trace $1$ is called a density matrix.
We denote by $S(\mathbb{C}^{2\otimes N})$ the set of all density matrices of order $2^N$. 
An element of $S(\mathbb{C}^{2\otimes N})$ is called a quantum state.
We also use a complex vector $|\psi\rangle\in\mathbb{C}^{2\otimes N}$ for representing a pure quantum state $|\psi\rangle\langle\psi|\in S(\mathbb{C}^{2\otimes N})$.

Let $M=\sum_{\bm{x},\bm{y}\in\{0,1\}^N}m_{\bm{x},\bm{y}}|x_1\rangle\langle y_1|\otimes\cdots\otimes|x_N\rangle\langle y_N|$ be a square matrix with $m_{\bm{x},\bm{y}}\in\mathbb{C}$.
For an integer $p\in [N]$, define a map ${\rm Tr}_p:S(\mathbb{C}^{2\otimes N})\rightarrow S(\mathbb{C}^{2\otimes (N-1)})$ as 
\begin{align*}
{\rm Tr}_p(M)\coloneqq&\sum_{\bm{x},\bm{y}\in\{0,1\}^N}m_{\bm{x},\bm{y}}\cdot {\rm Tr}(|x_p\rangle\langle y_p|)|x_1\rangle\langle y_1|\otimes\\
&\cdots\otimes|x_{p-1}\rangle\langle y_{p-1}|\otimes|x_{p+1}\rangle\langle y_{p+1}|\otimes\\
&\cdots\otimes|x_{N}\rangle\langle y_{N}|.
\end{align*}
The map ${\rm Tr}_p$ is called a partial trace.


\section{Deletion Error-Correcting Codes\\ and Main Contribution}\label{delsection}

Recall that in classical coding theory, for an integer $1\leq t<N$, a $t$-deletion error is defined as a map from a bit sequence of length $N$ to its subsequence of length $N-t$.
We define a $t$-deletion error in the terms of quantum coding theory.

\begin{defi}[Deletion Error $D_P$] Let $1\leq t<N$ be an integer.
For a set $P=\{p_1,\dots,p_t\}\subset [N]$ with $p_1<\dots<p_t$, define a map $D_P:S(\mathbb{C}^{2\otimes N})\rightarrow S(\mathbb{C}^{2\otimes (N-t)})$ as 
\begin{align*}
D_P(\rho)\coloneqq{\rm Tr}_{p_1}\circ\dots\circ{\rm Tr}_{p_t}(\rho),
\end{align*}
where $\rho\in S(\mathbb{C}^{2\otimes N})$ is a quantum state.
Here the symbol $\circ$ indicates the composition of maps.
We call the map $D_P$ a $t$-deletion error with the deletion position $P$.
The cardinality $|P|$ is called the number of deletions for $D_P$.
\end{defi}

In this paper, an integer $1\leq t<N$ is fixed and denotes a number of deletions, and a set $P\subset [N]$ denotes the deletion position satisfying $|P|=t$.

\begin{defi}[Deletion Error-Correcting Code]\label{DECC}
We call an image of ${\rm Enc}$ an $[N,K]$ $t$-deletion error-correcting code if the following conditions hold.
\begin{itemize}
\item There exists a map ${\rm Enc}:\mathbb{C}^{2\otimes K}\rightarrow\mathbb{C}^{2\otimes N}$
 defined as the composition of two maps ${\rm enc}\circ{\rm pad}^{N,K}$,
 where the map ${\rm pad}^{N,K}:\mathbb{C}^{2\otimes K}\rightarrow\mathbb{C}^{2\otimes N}$ is defined by ${\rm pad}^{N,K}(|\psi\rangle)\coloneqq|\psi\rangle\otimes|0\rangle\otimes\dots\otimes|0\rangle$, and the map ${\rm enc}:\mathbb{C}^{2\otimes N}\rightarrow\mathbb{C}^{2\otimes N}$ is a unitary transformation acting on $\mathbb{C}^{2\otimes N}$.
\item There exists a map ${\rm Dec}:S(\mathbb{C}^{2\otimes(N-t)})\rightarrow \mathbb{C}^{2\otimes K}$ defined by the operations allowed by quantum mechanics, 
and
\begin{align*}
{\rm Dec}\circ D_P\circ{\rm Enc}(|\psi\rangle)=|\psi\rangle
\end{align*}
holds for any quantum state $|\psi\rangle\in\mathbb{C}^{2\otimes K}$ and any deletion position $P\subset [N]$ satisfying $|P|=t$.
\end{itemize}
In other words, there exist an encoder ${\rm Enc}$ and a decoder ${\rm Dec}$ that correct any $t$-deletion errors.
\end{defi}

Note that a $t$-deletion error-correcting code is an $s$-deletion error-correcting code for any positive integer $s\leq t$. 

The following Definition \ref{del} is the conditions for constructing our codes.
Note that for binomial coefficients, if $w<0$ or $N<w$, we define $\binom{N}{w}\coloneqq0$.

\begin{defi}[Conditions (D1), (D2), and (D3)]\label{del}
For two non-empty sets $A,B\subset\{0,1,\dots,N\}$ and a map $f:A\cup B\rightarrow\mathbb{C}$,
define three conditions (D1), (D2), and (D3) as follows:\smallskip
\begin{itemize}
\setlength{\leftskip}{0.4cm}
\item[\textit{(D1)}]
$\displaystyle \sum_{w\in A}|f(w)|^2{\binom{N}{w}}=\sum_{w\in B}|f(w)|^2{\binom{N}{w}}=1$.\smallskip
\item[\textit{(D2)}] For any integer $0\leq k\leq t$,
\begin{align*}
\sum_{w\in A}|f(w)|^2{\binom{N-t}{w-k}}=\sum_{w\in B}|f(w)|^2{\binom{N-t}{w-k}}\neq0.
\end{align*}
\item[\textit{(D3)}] For any integers $w_1,w_2\in A\cup B$,
\begin{align*}
w_1\neq w_2~\Longrightarrow ~|w_1-w_2|>t.
\end{align*}
\end{itemize}
\end{defi}

The following Theorem \ref{main} is the main theorem of this paper and describes a new construction method for quantum deletion error-correcting codes.
\begin{thm}\label{main}
Let $A,B\subset\{0,1\}^N$ be non-empty sets with $A\cap B=\varnothing$ and $f:A\cup B\rightarrow\mathbb{C}$ be a map satisfying the conditions (D1), (D2), and (D3).
Then the code $Q_{A,B}^f$ is an $[N,1]$ $t$-deletion error-correcting code with the encoder ${\rm Enc}_{A,B}^f$ and the decoder ${\rm Dec}_{A,B}^f$.
\end{thm}

Here, the notations $Q_{A,B}^f$, ${\rm Enc}_{A,B}^f$, and ${\rm Dec}_{A,B}^f$ above are defined in the next section.


\section{Proof of The Main Theorem}\label{codeconst}

In this section, we shall give the proof of Theorem \ref{main}.
The proofs of Lemmas in this section are given in later appendices.

\begin{defi}[Encoder ${\rm Enc}_{A,B}^f$ and Code $Q_{A,B}^f$]\label{Enc}
Let $A,B\subset\{0,1,\dots,N\}$ be non-empty sets with $A\cap B=\varnothing$ and $f:A\cup B\rightarrow\mathbb{C}$ be a map satisfying the conditions (D1), (D2), and (D3).
Let us define an encoder as a linear map ${\rm Enc}_{A,B}^f:\mathbb{C}^2\rightarrow \mathbb{C}^{2\otimes N}$.
For a quantum state $|\psi\rangle=\alpha|0\rangle+\beta|1\rangle\in\mathbb{C}^2$, 
${\rm Enc}_{A,B}^f$ maps the state $|\psi\rangle$ to the following state $|\Psi\rangle$,
\begin{align}
|\Psi\rangle\coloneqq\!\sum_{\substack{\bm{x}\in\{0,1\}^{N}\\{\rm wt}(\bm{x})\in A}}\!\!\alpha f({\rm wt}(\bm{x}))|\bm{x}\rangle
+\!\sum_{\substack{\bm{y}\in\{0,1\}^{N}\\{\rm wt}(\bm{y})\in B}}\!\!\beta f({\rm wt}(\bm{y}))|\bm{y}\rangle.\label{psi2}
\end{align}
Set $Q_{A,B}^f$ as the image of ${\rm Enc}_{A,B}^f$, i.e.,
\begin{align*}
Q_{A,B}^f\coloneqq\{{\rm Enc}_{A,B}^f(|\psi\rangle)\mid |\psi\rangle\in\mathbb{C}^2,|\psi\rangle\langle\psi|\in S(\mathbb{C}^2)\}.
\end{align*}
\end{defi}

The map ${\rm Enc}_{A,B}^f$ can be written as a form ${\rm enc}\circ{\rm pad}^{N,K}$ with some unitary transformation ${\rm enc}$, thus it satisfies the condition of Definition \ref{DECC}.
Note that for the vector $|\Psi\rangle$ defined by Equation (\ref{psi2}),
\begin{align*}
\langle\Psi|\Psi\rangle&=\sum_{\substack{\bm{x}\in\{0,1\}^{N}\\{\rm wt}(\bm{x})\in A}}|\alpha f({\rm wt}(\bm{x}))|^2
+\sum_{\substack{\bm{y}\in\{0,1\}^{N}\\{\rm wt}(\bm{y})\in B}}|\beta f({\rm wt}(\bm{y}))|^2\\
&=\sum_{w\in A}|\alpha f(w)|^2{\binom{N}{w}}+\sum_{w\in B}|\beta f(w)|^2{\binom{N}{w}}\\
&=|\alpha|^2+|\beta|^2\\
&=1
\end{align*}
holds by the condition (D1).
Therefore, $|\Psi\rangle$ is a quantum state.

A quantum state is permutation-invariant if the state is invariant under position permutations.
A quantum code is permutation-invariant if any state of the code is permutation-invariant.
The term permutation-invariant is also called as PI.
The code $Q_{A,B}^f$ is a PI code.
The following Lemma \ref{lem2} describes the state after a deletion error for a PI state.

\begin{lem}\label{lem2}
Let $|\Psi\rangle$ be a pure PI state with
\begin{align}
|\Psi\rangle&\coloneqq\sum_{\bm{x}\in\{0,1\}^N}c({\rm wt}(\bm{x}))|\bm{x}\rangle,\label{psi1}
\end{align}
for a map $c:\{0,1,\dots,N\}\rightarrow\mathbb{C}$.
For an integer $0\leq k\leq t$, 
\begin{align}
|\Psi_k\rangle&\coloneqq\sum_{\bm{x}\in\{0,1\}^{N-t}}c({{\rm wt}(\bm{x})+k})|\bm{x}\rangle.\label{psik}
\end{align}
Then, for any deletion position $P\subset [N]$ satisfying $|P|=t$,
\begin{align*}
D_P(|\Psi\rangle\langle\Psi|)=\sum_{k=0}^t{\binom{\,t\,}{k}}|\Psi_k\rangle\langle\Psi_k|.
\end{align*}
\end{lem}

Note that Equation (\ref{psi2}) is obtained in Equation (\ref{psi1}) by setting
\begin{align}\label{cons}
c(w)=
\begin{cases}
\alpha f(w)&w\in A,\\
\beta f(w)&w\in B,\\
0&{\rm otherwise}.
\end{cases}
\end{align}
Equation (\ref{psik}) in our codes can be expressed in good form by the following Lemma \ref{lem4}.
The conditions (D2) and (D3) can be considered as an adaptation of the Knill and Laflamme conditions\cite{knill1997} to PI codes for deletion errors.

\begin{lem}\label{lem4}
Let $A,B\subset\{0,1\}^N$ be non-empty sets with $A\cap B=\varnothing$ and $f:A\cup B\rightarrow\mathbb{C}$ be a map satisfying the conditions (D2) and (D3).
Then for any integer $0\leq k\leq t$, there exist a real number $l_k\neq0$ and vectors $|u_1^k\rangle, |u_2^k\rangle\in\mathbb{C}^{2\otimes(N-t)}$ that satisfy the followings:
\begin{itemize}
\item For a vector $|\Psi_k\rangle$ defined by Equations (\ref{psik}) and (\ref{cons}),
\begin{align*}
|\Psi_k\rangle=l_k(\alpha|u_1^k\rangle+\beta|u_2^k\rangle).
\end{align*}
\item For integers $k_1,k_2\in \{0,1,\dots, t\}$ and $b_1,b_2\in \{1,2\}$,
\begin{align*}
\langle u_{b_1}^{k_1}|u_{b_2}^{k_2}\rangle=
\begin{cases}
1&(k_1,b_1)=(k_2,b_2),\\
0&(k_1,b_1)\neq(k_2,b_2).
\end{cases}
\end{align*}
\end{itemize}
\end{lem}

A set $\mathbb{P}\coloneqq\{P_k\}$ of projection matrices of order $2^N$ is called a projective measurement if and only if $\sum_kP_{k}=\mathbb{I}$, where $k$ is an index and $\mathbb{I}$ is the identity matrix of order $2^N$.
If the quantum state immediately before the measurement is $\rho\in S(\mathbb{C}^{2\otimes N})$ then the probability that outcome $k$ occurs is given by ${\rm Tr}(P_k\rho)$, and the quantum state $\rho'$ after the measurement is $\rho'\coloneqq\frac{P_k\rho P_k}{{\rm Tr}(P_k\rho)}$.

\begin{defi}[Set $\mathbb{P}_{A,B}$]\label{pab}
Let $A,B\subset\{0,1\}^N$ be sets. 
For an integer $0\leq k\leq t$, suppose that 
\begin{align*}
\displaystyle W_k\coloneqq\{\bm{x}\in\{0,1\}^{N-t}\mid {\rm wt}({\bm{x}})+k\in A\cup B\}.
\end{align*}
Then we define a set $\mathbb{P}_{A,B}\coloneqq\{P_0,P_1,\dots,P_t,P_\varnothing\}$ of projection matrices, where
\begin{align*}
P_k\coloneqq
\begin{cases}
\sum_{\bm{x}\in W_k}|\bm{x}\rangle\langle\bm{x}|&k\in\{0,1,\dots,t\},\medskip\\
\mathbb{I}-\sum_{k=0}^tP_k&k=\varnothing.
\end{cases}
\end{align*}
\end{defi}

If non-empty sets $A,B\subset\{0,1\}^N$ satisfy the condition (D3), the set $\mathbb{P}_{A,B}$ is clearly a projective measurement.
The following Lemma \ref{lem6} shows the results of the projective measurement $\mathbb{P}_{A,B}$ under the state after a deletion error in our code $Q_{A,B}^f$.

\begin{lem}\label{lem6}
Let $A,B\subset\{0,1\}^N$ be non-empty sets with $A\cap B=\varnothing$ and $f:A\cup B\rightarrow\mathbb{C}$ be a map satisfying the conditions (D1), (D2), and (D3).
Let $|\Psi\rangle$ and $|\Psi_k\rangle$ for an integer $0\leq k\leq t$ be defined by Equations (\ref{psi1}), (\ref{psik}) and (\ref{cons}).
If we perform the projective measurement $\mathbb{P}_{A,B}$ under the quantum state $D_P(|\Psi\rangle\langle\Psi|)\in S(\mathbb{C}^{2\otimes (N-t)})$ for any deletion position $P\subset [N]$, the probability $p(k)$ of getting outcome $k\in\{0,1,\dots,t,\varnothing\}$ is
\begin{align*}
p(k)=
\begin{cases}
\binom{\,t\,}{k}{l_k}^2&k\in\{0,1,\dots,t\},\smallskip\\
0&k=\varnothing.
\end{cases}
\end{align*}
When the outcome $k\in\{0,1,\dots,t\}$ is obtained,
 the quantum state $\rho(k)\in S(\mathbb{C}^{2\otimes (N-t)})$ after the measurement is 
\begin{align*}
\rho(k)=\frac{1}{{l_k}^2}|\Psi_k\rangle\langle\Psi_k|.
\end{align*}
\end{lem}

\begin{defi}[Error-Correcting Operator $U_k$]\label{ope}
Suppose the assumptions of Lemma \ref{lem4} are satisfied.
Then, for integers $k\in\{0,1,\dots, t\}$ and $m\in\{1,2\}$, we can choose a unitary matrix $U_k$ whose $m$th row is $\langle u_{m}^k|$.
We call the matrix $U_k$ an error-correcting operator.
The $m$th row of $U_k$ for $3\leq m\leq 2^{N-t}$ is also denoted as $\langle u_{m}^k|$.
\end{defi}

\begin{defi}[Decoding Algorithm ${\rm Dec}_{A,B}^f$]\label{Deco}
Let $A,B\subset\{0,1\}^N$ be non-empty sets with $A\cap B=\varnothing$ and $f:A\cup B\rightarrow\mathbb{C}$ be a map satisfying the conditions (D1), (D2), and (D3).
Define a decoder ${\rm Dec}_{A,B}^f
$ as a map from $\rho\in S(\mathbb{C}^{2\otimes (N-t)})$ to $\sigma\in\mathbb{C}^{2}$ constructed by the following steps:

\begin{itemize}
  \setlength{\leftskip}{0.7cm}
\item[\textit{(Step 1)}] Perform the projective measurement $\mathbb{P}_{A,B}$ under the quantum state $\rho$.
Assume that the outcome is $k$ and that the state after the measurement is $\rho(k)$.
\item[\textit{(Step 2)}] Let $\tilde{\rho}\coloneqq U_k\rho(k)U_k^\dag$.
Here $U_k$ is the error-correcting operator.
\item[\textit{(Step 3)}] At last, return $\sigma\coloneqq\underbrace{{\rm Tr}_1\circ\dots\circ{\rm Tr}_1}_{(N-t-1) \textit{ times}}(\tilde{\rho})$.
\end{itemize}
\end{defi}

We now describe the proof of the main theorem.

\begin{proof}[Proof of Theorem \ref{main}]
Set $|\Psi\rangle:={\rm Enc}_{A,B}^f(|\psi\rangle)$ for a pure quantum state $|\psi\rangle\in\mathbb{C}^2$.
For an integer $k\in\{0,1,\dots ,t\}$ and integers $i,j\in [2^{N-t}]$, we denote the $(i,j)$ element of the matrix $U_k\left(\frac{1}{{l_k}^2}|\Psi_k\rangle\langle\Psi_k|\right)U_k^\dag$ by $u_{k}(i,j)$.
By Lemma \ref{lem4}, 
we have
\begin{align}
u_{k}(i,j)&=\langle u_i^k|\left(\frac{1}{{l_k}^2}|\Psi_k\rangle\langle\Psi_k|\right)|u_j^k\rangle\nonumber\\
&=\langle u_i^k|(\alpha|u_1^k\rangle+\beta|u_2^k\rangle)(\overline{\alpha}\langle u_1^k|+\overline{\beta}\langle u_2^k|)|u_j^k\rangle\nonumber\\
&=(\alpha\langle u_i^k|u_1^k\rangle+\beta\langle u_i^k|u_2^k\rangle)(\overline{\alpha}\langle u_1^k|u_j^k\rangle+\overline{\beta}\langle u_2^k|u_j^k\rangle)\nonumber\\
&=
\begin{cases}
|\alpha|^2&(i,j)=(1,1),\\
\alpha\overline{\beta}&(i,j)=(1,2),\\
\overline{\alpha}\beta&(i,j)=(2,1),\\
|\beta|^2&(i,j)=(2,2),\\
0&{\rm otherwise}.
\end{cases}
\label{sigma}
\end{align}

By Lemmas \ref{lem2}, \ref{lem6}, Definition \ref{Deco}, and Equation (\ref{sigma}),
\begin{align*}
&{\rm Dec}_{A,B}^f\circ D_P\circ{\rm Enc}_{A,B}^f(|\psi\rangle\langle\psi|)\\
&~~~~={\rm Dec}_{A,B}^f\circ D_P(|\Psi\rangle\langle\Psi|)\\
&~~~~={\rm Dec}_{A,B}^f\left(\sum_{k=0}^t{\binom{\,t\,}{k}}|\Psi_k\rangle\langle\Psi_k|\right)\\
&~~~~={\rm Tr}_1\circ\dots\circ{\rm Tr}_1\left(U_k\left(\frac{1}{{l_k}^2}|\Psi_k\rangle\langle\Psi_k|\right)U_k^\dag\right)\\
&~~~~={\rm Tr}_1\circ\dots\circ{\rm Tr}_1(|0\rangle\langle0|\otimes\dots\otimes|0\rangle\langle0|\otimes|\psi\rangle\langle\psi|)\\
&~~~~=|\psi\rangle\langle\psi|
\end{align*}
holds for any pure quantum state $|\psi\rangle\in \mathbb{C}^2$ and any deletion position $P\in [N]$.
This is exactly the original quantum state.
\end{proof}


\section{Examples}\label{example}

By Theorem \ref{main}, we can construct a permutation-invariant quantum code for deletion errors from finding two non-empty sets $A,B\subset\{0,1,\dots,N\}$ with $A\cap B=\varnothing$ and a map $f:A\cup B\rightarrow\mathbb{C}$ that satisfy the three conditions (D1), (D2), and (D3).
we give two families of our codes in this section.

\subsection{Example 1 (Multiple-deletion error-correcting codes)}

First, we introduce a key combinatorial equation in giving the first example.

\begin{lem}\label{comb}
Let $n$ be a positive integer.
Then for all integers $0\leq t\leq n-1$,
\begin{align*}
\sum_{l=0}^n{\binom{n}{l}}l^t(-1)^l=0.
\end{align*}
\end{lem}

Lemma \ref{comb} can be easily shown by induction using the binomial identity $\sum_{l=0}^n\binom{n}{l}\binom{l}{t}(-1)^l=0$, which holds for any integer $0\leq t<n$.

The following Theorem \ref{ouy} gives quantum codes for deletion errors that have never been known before.
This is an interesting example that can be proved by good use of the combinatorial equation above.
Here, we fix an integer $1\leq t<N$.

\begin{thm}\label{ouy}
Let $g,n$ be integers and $u$ be a rational number with $g\geq t+1$, $n\geq t+1$, $u\coloneqq\frac{N}{gn}\geq1$.
Suppose that sets $A,B\subset\{0,1,\dots,N\}$ and a map $f:A\cup B\rightarrow\mathbb{C}$ are set as 
\begin{align*}
A&\coloneqq\{gl\mid0\leq l\leq n,l:{\rm even}\},\\
B&\coloneqq\{gl\mid0\leq l\leq n,l:{\rm odd}\},\\
f(gl)&\coloneqq\sqrt{\frac{\binom{n}{l}}{2^{n-1}\binom{gnu}{gl}}}.
\end{align*}
Then, the code $Q_{A,B}^f$ is an $[N,1]$ $t$-deletion error-correcting code.
\end{thm}

\begin{proof}
It is clear that $A\neq\varnothing$, $B\neq\varnothing$, $A\cap B=\varnothing$.
Hence, it is enough to prove the three conditions (D1), (D2), and (D3) hold by Theorem \ref{main}.

A simple calculation shows that 
\begin{align*}
\sum_{w\in A}|f(w)|^2{\binom{N}{w}}
=\frac{1}{2^{n-1}}\sum_{\substack{0\leq l\leq n\\l{\rm ~even}}}{\binom{n}{l}}
=1.
\end{align*}
Similarly, $\sum_{w\in B}|f(w)|^2\binom{N}{w}=1$. Therefore (D1) holds.

For an integer $0\leq k\leq t$, we obtain
\begin{align*}
& \sum_{w\in A}|f(w)|^2{\binom{N-t}{w-k}}-\sum_{w\in B}|f(w)|^2{\binom{N-t}{w-k}}\\
&~~~~=\sum_{l=0}^n\frac{\binom{n}{l}}{2^{n-1}\binom{gnu}{gl}}{\binom{gnu-t}{gl-k}}(-1)^l\\
&~~~~=0
\end{align*}
by the assumption $n\geq t+1$ and Lemma \ref{comb}.
Note that the ratio of binomial coefficients $\binom{gnu-t}{gl-k}/\binom{gnu}{gl}$ is a polynomial in $l$ of order $t$.
On the other hand, it is obvious that $\sum_{w\in A}|f(w)|^2\binom{N-t}{w-k}\neq0$. 
Therefore, (D2) holds.

It is clear that (D3) holds by the assumption $g\geq t+1$.
\end{proof}

The quantum code constructed by Theorem \ref{ouy} include some that are already known, but its tolerance to deletion errors was mentioned here for the first time.
This code is called a $(g,n,u)$-PI code in the terms of Ouyang\cite{Ouyang2014}.
Theorem \ref{ouy} claims that a $(g,n,u)$-PI code is a $t$-deletion error-correcting code if $g\geq t+1$, $n\geq t+1$, and $u\geq1$.
The smallest example is precisely Hagiwara's 4-qubit code that is a $(2,2,1)$-PI code\cite{Hagiwara20202}.
The following Fact \ref{Ouy2014} is a remarkable result about PI codes by Ouyang \cite{Ouyang2014}.

\begin{fact}
\label{Ouy2014}
For an integer $t\geq 1$, $(g,n,u)$-PI codes correct arbitrary $t$-qubit errors if $g=n=2t+1$ and $u\geq1$.
\end{fact}

Fact \ref{Ouy2014} and Theorem \ref{ouy} show that a $(2t+1,2t+1,u)$-PI code with an integer $t\geq1$ is $t$-qubit and $2t$-deletion error-correcting code.
This is the first example of codes that can correct both deletion errors and another type of errors.
The smallest example is precisely Ruskai's 9-qubit code that is a $(3,3,1)$-PI code\cite{ruskai2000}.
It was already known that this code can correct $1$-qubit errors, but it was shown here for the first time that it can  also correct $2$-deletion errors.

\subsection{Example 2 (Single-deletion error-correcting codes\cite{Shibayama2020})}

Here, we fix $t:=1$ and introduce $1$-deletion error-correcting codes.
The codes constructed by following Theorem \ref{shiba} are already known as examples\cite{Shibayama2020} of the code construction given by Nakayama and Hagiwara\cite{Nakayama20202}, but it is also one family of our codes.
\begin{thm}\label{shiba}
Suppose that two non-empty sets $A,B\subset\{0,1,\dots,N\}$ with $A\cap B=\varnothing$ satisfy followings:
\begin{itemize}
\item $w\in A\Rightarrow N-w\in A$,\smallskip
\item $w\in B\Rightarrow N-w\in B$,\smallskip
\item For any integers $w_1,w_2\in A\cup B$,
\begin{align*}
w_1\neq w_2~\Longrightarrow ~|w_1-w_2|>1.
\end{align*}
\end{itemize}
and a map $f:A\cup B\rightarrow\mathbb{C}$ are set as
\begin{align*}
f(w):=
\begin{cases}
\sqrt{\frac{1}{\sum_{w'\in A}\binom{N}{w'}}}&w\in A,\medskip\\
\sqrt{\frac{1}{\sum_{w'\in B}\binom{N}{w'}}}&w\in B.
\end{cases}
\end{align*}
Then, the code $Q_{A,B}^f$ is an $[N,1]$ $1$-deletion error-correcting code.
\end{thm}

\begin{proof}
It is clear that (D1) and (D3) hold by the assumptions.
Hence we show that (D2) holds. 
By the assumption, we have
\begin{gather*}
\sum_{w\in A}{\binom{N-1}{w-0}}=\sum_{w\in A}{\binom{N-1}{w-1}},\\
\sum_{w\in A}{\binom{N-1}{w-0}}+\sum_{w\in A}{\binom{N-1}{w-1}}=\sum_{w\in A}{\binom{N}{w}}.
\end{gather*}
Similarly, the same equations for $B$ are obtained.
Hence, 
\begin{align*}
&\sum_{w\in A}|f(w)|^2{\binom{N-t}{w-k}}-\sum_{w\in B}|f(w)|^2{\binom{N-t}{w-k}}\\
&~~~~~~
=\frac{\sum_{w\in A}\binom{N-1}{w-k}}{\sum_{w'\in A}\binom{N}{w'}}-\frac{\sum_{w\in B}\binom{N-1}{w-k}}{\sum_{w'\in B}\binom{N}{w'}}\\
&~~~~~~=\frac1{\,2\,}-\frac1{\,2\,}\\
&~~~~~~=0
\end{align*}
holds for any integer $0\leq k\leq 1$.
Therefore, (D2) holds.
\end{proof}


\section{Generalization}\label{generalize}

In this section, we discuss constructions of $[N,K]$ $t$-deletion error-correcting codes for any positive integer $K$.
Let $L$ be a positive integer and $A_0,A_1,\dots,$ $A_{L-1}\subset\{0,1,\dots,N\}$ be mutually disjoint non-empty sets and $f:\bigcup_{i=0}^{L-1}A_i\rightarrow\mathbb{C}$ be a map which satisfy the following three conditions:
\begin{itemize}
\setlength{\leftskip}{0.4cm}
\item[(D1)*] For any integer $i\in\{0,1,\dots,L-1\}$,
\begin{align*}
\sum_{w\in A_i}|f(w)|^2{\binom{N}{w}}=1.
\end{align*}
\item[(D2)*] For any integers $0\leq k\leq t$ and $i,j\in\{0,1,\dots,L-1\}$,
\begin{align*}
\sum_{w\in A_i}|f(w)|^2{\binom{N-t}{w-k}}=\sum_{w\in A_j}|f(w)|^2{\binom{N-t}{w-k}}\neq0.
\end{align*}
\item[(D3)*] For any integers $w_1,w_2\in \bigcup_{i=0}^{L-1}A_i$,
\begin{align*}
w_1\neq w_2~\Longrightarrow ~|w_1-w_2|>t.
\end{align*}
\end{itemize}
Let us define an encoder as a linear map ${\rm Enc}_{\{A_i\}}^f:\mathbb{C}^L\rightarrow \mathbb{C}^{2\otimes N}$.
For a quantum state $|\psi\rangle=\sum_{i=0}^{L-1}\textstyle\alpha_i|i\rangle\in\mathbb{C}^L$, where $|0\rangle,|1\rangle,\dots,|L-1\rangle$ is the standard orthogonal basis of $\mathbb{C}^L$, ${\rm Enc}_{\{A_i\}}^f$ maps the state $|\psi\rangle$ to the following state $|\Psi\rangle$,
\begin{align*}
|\Psi\rangle&\coloneqq\sum_{i=0}^{L-1}\textstyle\left(\sum_{\substack{\bm{x}\in\{0,1\}^{N}\\{\rm wt}(\bm{x})\in A_i}}\alpha_if({\rm wt}(\bm{x}))|\bm{x}\rangle\right).
\end{align*}
Note that this encoder is an extension of Definition \ref{Enc}.
We claim that the image of ${\rm Enc}_{\{A_i\}}^f$ is a $t$-deletion error-correcting code for an integer $1\leq t<N$.
We can use the same method as in Section \ref{codeconst} for the proof.
For the case $L=2^K$, we obtain a $[N,K]$ $t$-deletion error-correcting code.

Although we do not discuss it in detail here, we can construct $[N,K]$ single-deletion error-correcting codes with any integer $K\geq1$ by extending Theorem \ref{shiba}.
But no example of $[N,K]$ $t$-deletion error-correcting codes with $K\geq2$ and $t\geq2$ has been found to date.


\section{Conclusion}\label{conc}
This paper gave a construction of permutation-invariant quantum codes for deletion errors.
In particular, the codes given in Theorem \ref{ouy} contain the first example of quantum codes that can correct two or more deletion errors and the first example of codes that can correct both multiple-qubit errors and multiple-deletion errors.

\section*{Acknowledgment}
The research has been partly executed in response to support by KAKENHI 18H01435.


 \appendix
\def\thesection{Appendix~\Alph{section}}

  \section{Proof of Lemma \ref{lem2}}\label{prf2}
\begin{proof}
By Equations (\ref{psi1}) and (\ref{psik}), it is clear that 
\begin{align*}
|\Psi\rangle=\sum_{\bm{y}\in\{0,1\}^t}\left(|\bm{y}\rangle\otimes|\Psi_{{\rm wt}(\bm{y})}\rangle\right)
\end{align*}
for any integer $1\leq t<N$. By the permutation-invariance of $|\Psi\rangle$ and the definition of the partial trace,
\begin{align*}
D_P(|\Psi\rangle\langle\Psi|)&=\underbrace{{\rm Tr}_1\circ\dots\circ{\rm Tr}_1}_{t\textrm{ times}}(|\Psi\rangle\langle\Psi|)\\
&=\sum_{\bm{y}\in\{0,1\}^t}|\Psi_{{\rm wt}(\bm{y})}\rangle\langle\Psi_{{\rm wt}(\bm{y})}|\\
&=\sum_{k=0}^t{\binom{\,t\,}{k}}|\Psi_k\rangle\langle\Psi_k|
\end{align*}
holds for any deletion position $P\subset [N]$.
\end{proof}

  \section{Proof of Lemma \ref{lem4}}\label{prf4}
\begin{proof}
For an integer $0\leq k\leq t$, suppose that
\begin{align*}
|U_{1}^{k}\rangle&\coloneqq\sum_{w\in A}
\left(\textstyle
\sum_{\substack{\bm{x}\in\{0,1\}^{N-t}\\{\rm wt}(\bm{x})=w-k}}
f(w)|\bm{x}\rangle\right),\\
|U_{2}^{k}\rangle&\coloneqq\sum_{w\in B}
\left(\textstyle
\sum_{\substack{\bm{y}\in\{0,1\}^{N-t}\\{\rm wt}(\bm{y})=w-k}}
f(w)|\bm{y}\rangle\right).
\end{align*}
By the condition (D2), $\langle U_{1}^{k}|U_{1}^{k}\rangle=\langle U_{2}^{k}|U_{2}^{k}\rangle\neq0$.
Set $l_k\in\mathbb{R}$ and $|u_1^k\rangle,|u_2^k\rangle\in\mathbb{C}^{2\otimes(N-t)}$ as
\begin{align*}
l_k\coloneqq\sqrt{\langle U_{1}^{k}|U_{1}^k\rangle},~~|u_1^k\rangle\coloneqq\frac{|U_1^k\rangle}{l_k},~~|u_2^k\rangle\coloneqq\frac{|U_2^k\rangle}{l_k}.
\end{align*}
Then, $\langle u_{1}^{k}|u_{1}^{k}\rangle=\langle u_{2}^{k}|u_{2}^{k}\rangle=1$ holds.
Hence, we have
\begin{align*}
|\Psi_k\rangle&=\alpha|U_1^k\rangle+\beta|U_2^k\rangle=l_k(\alpha|u_1^k\rangle+\beta|u_2^k\rangle)
\end{align*}
by Equations (\ref{psik}) and (\ref{cons}),
In the case $(k_1,b_1)\neq(k_2,b_2)$, we obtain $\langle u_{b_1}^{k_1}|u_{b_2}^{k_2}\rangle=0$ by the condition (D3).
\end{proof}

  \section{Proof of Lemma \ref{lem6}}\label{prf6}
\begin{proof}
In the case $k\in\{0,1,\dots,t\}$, we have
\begin{align*}
p(k)&={\rm Tr}(P_kD_P(|\Psi\rangle\langle\Psi|))\\
& ={\rm Tr}\left(\sum_{\bm{x}\in W_k}|\bm{x}\rangle\langle\bm{x}|\sum_{k'=0}^t{\binom{\,t\,}{k'}}|\Psi_{k'}\rangle\langle\Psi_{k'}|\right)\\
& ={\rm Tr}\left(\binom{\,t\,}{k}|\Psi_k\rangle\langle\Psi_k|\right)\\
& ={\rm Tr}\left(\binom{\,t\,}{k}{l_k}^2(\alpha|u_1^k\rangle+\beta|u_2^k\rangle)(\overline{\alpha}\langle u_1^k|+\overline{\beta}\langle u_2^k|)\right)\\
& =\binom{\,t\,}{k}{l_k}^2\left(|\alpha|^2\langle u_1^k|u_1^k\rangle+|\beta|^2\langle u_2^k|u_2^k\rangle\right)\\
& =\binom{\,t\,}{k}{l_k}^2
\end{align*}
by Lemmas \ref{lem2} and \ref{lem4}.
In the case $k=\varnothing$, 
it is clear that $p(\varnothing)={\rm Tr}(P_\varnothing D_P(|\Psi\rangle\langle\Psi|))=0$.

Given that outcome $k\in\{0,1,\dots ,t\}$ occurred, by Lemma \ref{lem2}, the quantum state immediately after the measurement is 
\begin{align*}
\frac{P_kD_P(|\Psi\rangle\langle\Psi|)P_k}{{\rm Tr}(P_kD_P(|\Psi\rangle\langle\Psi|))}&=\frac{P_k\left(\sum_{k'=0}^t\binom{t}{k'}|\Psi_{k'}\rangle\langle\Psi_{k'}|\right)P_k}{\binom{t}{k}{l_k}^2}\\
&=\frac{\binom{t}{k}|\Psi_k\rangle\langle\Psi_k|}{\binom{t}{k}{l_k}^2}\\
&=\frac{1}{{l_k}^2}|\Psi_k\rangle\langle\Psi_k|.
\end{align*}
\end{proof}


\bibliography{bibtex}

\begin{thebibliography}{10}

\bibitem{Biamonte2017}
Jacob Biamonte, Peter Wittek, Nicola Pancotti, Patrick Rebentrost, Nathan
  Wiebe, and Seth Lloyd.
\newblock Quantum machine learning.
\newblock {\em Nature}, 549:195--202, 2017.

\bibitem{Robert1996}
A~Robert Calderbank and Peter~W Shor.
\newblock Good quantum error-correcting codes exist.
\newblock {\em Physical Review A}, 54(2):1098--1105, Aug 1996.

\bibitem{Fowler2012}
Austin~G Flowler, Matteo Mariantoni, John~M Martinis, and Andrew~N Cleland.
\newblock Surface codes: Towards practical large-scale quantum computation.
\newblock {\em Physical Review A}, 86(3):032324, Sep 2012.

\bibitem{Gottesman1997}
Daniel Gottesman.
\newblock Stabilizer codes and quantum error correction.
\newblock {\em arXiv preprint arXiv:9705052}, 1997.

\bibitem{Grassl1997}
Markus Grassl, Th~Beth, and Thomas Pellizzari.
\newblock Codes for the quantum erasure channel.
\newblock {\em Physical Review A}, 56(1):33--38, Jul 1997.

\bibitem{Hagiwara20202}
Manabu Hagiwara and Ayumu Nakayama.
\newblock A four-qubits code that is a quantum deletion error-correcting code
  with the optimal length.
\newblock {\em 2020 IEEE International Symposium on Information Theory (ISIT)},
  pages 1870--1874, 2020.

\bibitem{Hillery1999}
Mark Hillery, Vladim\'{\i}r Bu\ifmmode~\check{z}\else \v{z}\fi{}ek, and Andr\'e
  Berthiaume.
\newblock Quantum secret sharing.
\newblock {\em Physical Review A}, 59(3):1829--1834, Mar 1999.

\bibitem{Hughston1993}
Lane~P. Hughston, Richard Jozsa, and William~K. Wootters.
\newblock A complete classification of quantum ensembles having a given density
  matrix.
\newblock {\em Physics Letters A}, 183(1):14 -- 18, 1993.

\bibitem{knill1997}
Emanuel Knill and Raymond Laflamme.
\newblock Theory of quantum error-correcting codes.
\newblock {\em Physical Review A}, 55(2):900--911, Feb 1997.

\bibitem{Leahy2019}
Janet Leahy, Dave Touchette, and Penghui Yao.
\newblock Quantum insertion-deletion channels.
\newblock {\em arXiv preprint arXiv:1901.00984}, 2019.

\bibitem{Nakayama20201}
Ayumu Nakayama and Manabu Hagiwara.
\newblock The first quantum error-correcting code for single deletion errors.
\newblock {\em IEICE Communications Express}, 9(4):100--104, 2020.

\bibitem{Nakayama20202}
Ayumu Nakayama and Manabu Hagiwara.
\newblock Single quantum deletion error-correcting codes.
\newblock {\em 2020 International Symposium on Information Theory and Its
  Applications (ISITA)}, pages 329--333, 2020.

\bibitem{OC2019}
Y.~{Ouyang} and R.~{Chao}.
\newblock Permutation-invariant constant-excitation quantum codes for amplitude
  damping.
\newblock {\em IEEE Transactions on Information Theory}, 66(5):2921--2933,
  2020.

\bibitem{Ouyang2014}
Yingkai Ouyang.
\newblock Permutation-invariant quantum codes.
\newblock {\em Physical Review A}, 90(6):062317, Dec 2014.

\bibitem{Ouyang2017}
Yingkai Ouyang.
\newblock Permutation-invariant qudit codes from polynomials.
\newblock {\em Linear Algebra and its Applications}, 532:43 -- 59, 2017.

\bibitem{Ouyang2020}
Yingkai Ouyang.
\newblock Quantum storage in quantum ferromagnets.
\newblock {\em arXiv preprint arXiv:1904.01458}, 2020.

\bibitem{Ouyang2016}
Yingkai Ouyang and Joseph Fitzsimons.
\newblock Permutation-invariant codes encoding more than one qubit.
\newblock {\em Physical Review A}, 93(4):042340, Apr 2016.

\bibitem{OSM2019}
Yingkai Ouyang, Nathan Shettell, and Damian Markham.
\newblock Robust quantum metrology with explicit symmetric states.
\newblock {\em arXiv preprint arXiv:1908.02378}, 2019.

\bibitem{pollatsek2004}
Harriet Pollatsek and Mary~Beth Ruskai.
\newblock Permutationally invariant codes for quantum error correction.
\newblock {\em Linear Algebra and its Applications}, 392:255 -- 288, 2004.

\bibitem{ruskai2000}
Mary~Beth Ruskai.
\newblock Pauli exchange errors in quantum computation.
\newblock {\em Physical Review A}, 85(1):194--197, Jul 2000.

\bibitem{Shibayama2020}
Taro Shibayama.
\newblock New instances of quantum error-correcting codes for single deletion
  errors.
\newblock {\em 2020 International Symposium on Information Theory and Its
  Applications (ISITA)}, pages 334--338, 2020.

\bibitem{Peter1995}
Peter~W Shor.
\newblock Scheme for reducing decoherence in quantum computer memory.
\newblock {\em Physical Review A}, 52:R2493--R2496, Oct 1995.

\bibitem{smagloy2020}
Ilia Smagloy, Lorenz Welter, Antonia Wachter-Zeh, and Eitan Yaakobi.
\newblock Single-deletion single-substitution correcting codes.
\newblock {\em arXiv preprint arXiv:2005.09352}, 2020.

\bibitem{song2020}
Wentu Song, Nikita Polyanskii, Kui Cai, and Xuan He.
\newblock Systematic single-deletion multiple-substitution correcting codes.
\newblock {\em arXiv preprint arXiv:2006.11516}, 2020.

\bibitem{WWG2019}
Chunfeng Wu, Yimin Wang, Chu Guo, Yingkai Ouyang, Gangcheng Wang, and Xun-Li
  Feng.
\newblock Initializing a permutation-invariant quantum error-correction code.
\newblock {\em Physical Review A}, 99(1):012335, Jan 2019.

\end{thebibliography}

\end{document}